\documentclass[11pt,runningheads]{llncs}
\usepackage{amsmath}
\usepackage{amsfonts}
\usepackage{amssymb}
\usepackage{latexsym}
\usepackage{epsfig,graphicx}

\newcommand{\A}{{\mathcal{A}}}
\newcommand{\B}{{\mathcal{B}}}
\newcommand{\one}{\textbf{1}}
\newcommand{\E}{{\mathbb{E}}}
\newcommand{\var}{{\rm{Var}}}
\newcommand{\cov}{{\rm{Cov}}}

\newcommand{\triang}{{\rm \textsc{Triangles}}}
\newcommand{\dist}{{\rm \textsc{Dist}}}
\newcommand{\disj}{{\rm \textsc{Disj}}}
\newcommand{\ind}{{\rm \textsc{Index}}}

\newcommand{\x}{{\bf x}}
\newcommand{\y}{{\bf y}}

\newcommand{\G}{{\mathcal{G}}}

\title{How Hard is Counting Triangles in the Streaming Model}

\author{
Vladimir Braverman \inst{1}, Rafail Ostrovsky \inst{2} and Dan Vilenchik \inst{3}}

\institute{Department of Computer Science, Johns Hopkins University.
\\\email{vova@cs.jhu.edu}   \and Department of Computer Science, UCLA.\\ \email{rafail@cs.ucla.edu} \and Faculty of Mathematics and Computer Science, The Weizmann Institute, Israel.
\\\email{dan.vilenchik@weizmann.ac.il}}

\begin{document}

\maketitle

\begin{abstract}
The problem of (approximately) counting the number of triangles in a graph is one of the basic problems in graph theory. In this paper we study the problem in the streaming model. We study the amount of memory required by a randomized algorithm to solve this problem. In case the algorithm is allowed one pass over the stream, we present a best possible lower bound of $\Omega(m)$ for graphs $G$ with $m$ edges on $n$ vertices. If a constant number of passes is allowed, we show a lower bound of $\Omega(m/T)$, $T$ the number of triangles. We match, in some sense, this lower bound with a 2-pass $O(m/T^{1/3})$-memory algorithm that solves the problem of distinguishing graphs with no triangles from graphs with at least $T$ triangles. We present a new graph parameter $\rho(G)$ -- the triangle density, and conjecture that the space complexity of the triangles problem is $\Omega(m/\rho(G))$. We match this by a second algorithm that solves the distinguishing problem using $O(m/\rho(G))$-memory.
\end{abstract}

\newpage

\section{Introduction}
Counting the number of triangles in a graph $G=(V,E)$ is one of the basic algorithmic questions in graph theory.
From a theoretical point of view, the key question is to determine the time and space complexity of the problem. The brute-force approach enumerates all possible triples of nodes (taking $O(n^3)$ time, $n$ is the number of vertices in $G$). The algorithms with the lowest time complexity for counting triangles rely on fast matrix multiplication. The asymptotically fastest algorithm to date
is $O(n^{2.376})$ \cite{CoppersmithWinograd}. An algorithm that runs in time $O(m^{1.41})$ with space complexity $\Theta(n^2)$ is given in
\cite{AlonYuster} ($m$ is the number of edges in $G$). In more practical applications, the number of triangles is a frequently used network statistic in the exponential random graph model \cite{RinaldoFienberg,FrankStrauss}, and naturally
appears in models of real-world network evolution \cite{LeskovecBackstrom}, and web applications \cite{BecchettiBoldi,EckmannMoses}. In the context of social networks, triangles have a natural interpretation: friends of friends tend to
be friends \cite{WassermanFaust}, and this can be used in link recommendation/friend suggestion \cite{TsourakakisDrineas}.

The memory restrictions when dealing with huge graphs led to consider the streaming model: The edges of the graph come down the stream, and the algorithm processes each edge as it comes in an on-line fashion (once it moves down the stream, it cannot access it again). The algorithm is allowed ideally one pass (or a limited number of passes) over the stream. The parameter of interest is the amount of memory that the algorithm uses to solve the problem.
\begin{definition} $\triang(c)$ is the problem of approximating the number of triangles in the input graph within a multiplicative factor of 0.9, with probability at least $2/3$, using at most $c$ passes over the data stream.
\end{definition}
The choice of constants $0.9$ and $2/3$ in the definition is for the sake of clear and simple presentation. One can take both the approximation rate and success probability to be parameters of the problem.

Currently, no non-trivial algorithms are known to solve $\triang(c)$ when $c$ is constant (by trivial we mean an algorithm that uses $\Theta(m)$ memory, $m$ the number of edges, which is asymptotically the same as storing all graph edges). All existing approximation algorithms receive $T_3$ (the number of triangles in the input graph) as part of their input. Obviously, $T_3$ cannot be part of the input. One way to implement such an algorithm is by ``guessing" the correct value of $T_3$ and verifying the guess. This translates into $\Theta(\log n)$ passes over the stream (in form of a binary search for example). Let us mention some of the known results. One approach is a sampling approach. The algorithm suggested in~\cite{Buriol:2006} samples in the first pass $s$ random pairs $(e,v)$ of an edge $e=(u,w)$ and a vertex $v$, and stores them. Then in the second pass checks for every pair $(e,v)$ whether $(u,v,w)$ is a triangle. The total number of triangles is estimated as a function of the number of pairs $(e,v)$ that formed a triangle. The number of samples $s$ is proportional to $(T_0+T_1+T_2)/T_3$ (where $T_i = \#$ of vertex triples in the graph spanning exactly $i$ edges). They show a more sophisticated sampling algorithm which uses $(T_1+T_2)/T_3$ samples.  A different approach reduces the problem of approximating the number of triangles to the problem of estimating the frequency moments of the data stream, using the Alon-Matias-Szegedy (AMS) algorithm~\cite{AMS99}. This approach was presented in~\cite{Bar-Yossef:2002} for the first time. The algorithm in~\cite{Bar-Yossef:2002} uses $T_1,T_2,T_3$ to compute the appropriate parameters with which to run AMS. The space complexity of this algorithm is proportional to $((T_1+T_2)/T_3)^3$. In another work~\cite{Jowhari05newstreaming}, the algorithm uses $m,C_4,C_6,T_3$ to compute the appropriate parameters to AMS ($C_i$ is the number of $i$-cycles in the graph). The space complexity of that algorithm is $(m^3 + mC_4 + C_6 + T_3^2)/T_3^2$.

Observe that all aforementioned algorithms assume other non-trivial graph-parameters to be part of their input as well. Another disadvantage that the aforementioned algorithms share, is the fact that their space complexity depends on parameters that are not necessarily indicative of the number of triangles in the graph. For example, the parameter $T_2$ may have little to do with the number of triangles $T_3$ in some graphs. Consider the graph whose vertex set $V=A_0 \cup A_1 \cup A_2$, each $A_i$ of size $n/3$. The edge set $E$ is the complete bi-partite graph on $A_0,A_1$ and on $A_1,A_2$. Clearly, $G=(V,E)$ has no triangles, so $T_3=0$, but has $T_2=\Theta(n^3)$ paths of length two.  In light of what we've just said, two interesting questions arise

\begin{description}
\item[Question 1:] Determine the space complexity of $\triang(1)$ and \\ $\triang(O(1))$.
\end{description}

\begin{description}
\item[Question 2:] Is there an algorithm that solves $\triang(c)$, whose space complexity depends only on the number of edges and $T_3$?
\end{description}

Bar-Yossef~et.~al.~\cite{Bar-Yossef:2002} showed that the space complexity required to solve $\triang(1)$ is $\Omega(n^2)$ (throughout we disregard the memory it takes to represent a single graph vertex). Specifically, they showed that every one-pass $0.5$-approximation algorithm that succeeds with probability $0.99$, is as good as the trivial algorithm that stores all edges and exhaustively computes the number of triangles. While the lower bound determines the space complexity in the worst case, it is informative to study more refined notions of  ``worst case". For example, what is the space complexity of $\triang(1)$ when the graph has exactly $m$ edges, or at least $T$ triangles. What is the space complexity of $\triang(2)$, when two passes over the stream are allowed, rather than one. In those cases the lower bound in~\cite{Bar-Yossef:2002} is irrelevant.

\subsection{Our Results}

In this paper we show that the space complexity of any algorithm that solves $\triang(1)$  (i.e. in \emph{one} pass) is $\Omega(m)$. This lower bound is asymptotically \textbf{tight}, since the trivial algorithm that stores all edges of the graph uses that much memory. Furthermore, the lower bound is true even when assuming that the graph has $T_3=O(n)$ triangles. Clearly one cannot expect this to be the case for every value of $T_3$, since when $T_3 = \Theta(n^3)$ for example, a straightforward sampling algorithms solves $\triang(1)$ using $O(1)$-space. Formally,

\begin{theorem}\label{thm:LowerBound1} $\exists c_1,c_2 > 0$ s.t. the space complexity of $\triang(1)$ is $\Omega(m)$, when the input is an $n$-vertex graph with $m \in [c_1n,c_2n^2]$ edges. Furthermore, this lower bound is true even if the graph has as many as $0.99n$ triangles.
\end{theorem}

Theorem \ref{thm:LowerBound1} extends the aforementioned result in~\cite{Bar-Yossef:2002} in two aspects: the number of edges is asymptotically the entire range (compared with $\Theta(n^2)$ in~\cite{Bar-Yossef:2002}). The graph may contain as many as a linear number of triangles (compared with one triangle in~\cite{Bar-Yossef:2002}). In addition, our proof technique is conceptually and technically simpler.

Improving upon the currently best known lower bound of $\Omega(n/T_3)$ ($n$ is the number of vertices) for $\triang(O(1))$~\cite{Jowhari05newstreaming}, we show that:
\begin{theorem}\label{thm:LowerBound2}  The space complexity of $\triang(O(1))$ for input graphs with $m$ edges and $T_3$ triangles is $\Omega(m/\max\{T_3,1\})$.
\end{theorem}

Turing to the algorithmic part, the lower bound for $\triang(1)$ is asymptotically tight. Giving a non-trivial upper bound for $\triang(O(1))$ seems to be beyond the reach of current algorithmic techniques. As we already mentioned, all current state-of-the-art approximation algorithms require a super-constant number of passes (regardless of the space complexity). Hence, we start with a softer notion of approximation, in the spirit of property testing.

\begin{definition} $\dist(c)$ is the following problem. Given two graph families: $\G_1$ consisting of triangle-free graphs, $\G_2$ consisting of graphs with at least $T$ triangles, and an input graph $G \in \G_1 \cup \G_2$, decide whether $G \in \G_1$ or $G \in \G_2$ with probability at least $2/3$, using at most $c$ passes over the input.
\end{definition}
The same lower bounds that we derive for $\triang(c)$ are true for $\dist(c)$ as well (therefore there is nothing interesting to say algorithmically about $\dist(1)$). None of the aforementioned approximation algorithms solve $\dist(O(1))$, since they require additional parameters that are not available to the algorithm (for example, $T_2$, the number of triples spanning exactly two edges). We describe an algorithm that solves $\dist(2)$ using $O(m/T^{1/3})$ bits of memory. This answers Question 2 above, for the problem $\dist$.

We now turn to describe in details our algorithm, and formally state the relevant theorem. We assume that the parameter $T$ is known to the algorithm (as it is part of the problem definition).

\medskip

\begin{figure*}[!htp]
\begin{center}
\fbox{\normalsize
\begin{minipage}{\textwidth}{\textbf{Algorithm $A$}}
\begin{description}
  \item[Output:] `1' iff a triangle was found.
  \item[Pass 1]
  \item[(a)] Set $m' = 6m/T^{1/3}$, and $p = m'/m$.
  \item[(b)]  Store every edge $e$ with probability $p$. If more than $5m'$ edges are stored, output FAIL.
  \item[(c)]  Let $H$ be the graph stored by the algorithm at the end of $\rm\textbf{(b)}$. Search for a triangle in $H$, if found output 1.
  \item[Pass 2] For every edge $e$, check whether $e$ completes a triangle in $H$. Output 1 iff such edge exists.
\end{description}
\end{minipage}
}
\end{center}
\end{figure*}
\begin{theorem}\label{thm:alg1} For $T \geq 216$, Algorithm $A$ solves $\dist(2)$ using at most $30m/T^{1/3}$ bits of memory.
\end{theorem}

\noindent \textbf{Remark 1.} When $T =\omega(1)$, our algorithm solves $\dist(2)$ using \emph{sub-linear} space. Also, our lower bound on the space complexity of $\dist(1)$, together with Algorithm A for $\dist(2)$, imply a  space complexity separation result between one-pass and two-passes. For example, $\dist(2)$ can be solved in space $O(m/n)=o(m)$ for graphs with $T_3=n/2$ triangles and $m$ edges, while $\dist(1)$ requires $\Omega(m)$-space for such graphs.

\noindent \textbf{Remark 2.} Algorithm $A$ assumes $m$ is given. This assumption is done only for the sake of clear and simple presentation, and can be easily removed: The algorithm ``guesses" an initial value for $m$, say $m_1=1$. This value is used to define $p$ for the first $m_1$ edges. If the number of edges exceeds that guess, then the algorithm sets $m_2=2m_1$, and updates $p$ accordingly for the next $m_2$ edges. Every time guess $i$ is exceeded, the algorithm sets $m_{i+1}=2m_i$. The last interval will consist of the last $m/2$ edges. Edges are still stored independently of each other, and in expectation twice as many edges are stored. Storing more edges may only help the algorithm (while not changing the asymptotic space complexity). Hence the same analysis that we have for $A$ goes through with this additional procedure.

\subsection{A new graph parameter}\label{sec:NewGraphParam}
While the bound given in Theorem \ref{thm:LowerBound1} for $\triang(1)$ is asymptotically tight, we suspect that the bound in Theorem \ref{thm:LowerBound2} for $\triang(O(1))$ is not tight, and conjecture a tight bound instead. Define the \textbf{triangle density} of a graph $G$, $\rho(G)$, to be the number of vertices that belong to some triangle in $G$. It is easy to see that $\binom{(6T)^{1/3}}{3} \leq \rho(G) \leq 3T$ (a clique or $T$ disjoint triangles).

\medskip

\noindent \textbf{Conjecture:} The space complexity of $\triang(O(1))$ and $\dist(O(1))$ is $\Omega(m/\rho(G))$.

\medskip

The lower bound in Theorem \ref{thm:LowerBound2} is consistent with the case $\rho(G)=\Theta(T)$, and Algorithm $A$ is consistent with $\rho(G)=\Theta(T^{1/3})$. We describe a second algorithm  that solves $\dist(2)$ using $O(m/\rho(G))$ space, thus showing that one cannot hope for a tighter bound than the one stated in the conjecture.  A formal description of the algorithm, a proof of correctness and analysis of its space complexity is given in Section \ref{apnd:A2}.

\subsection{Techniques}
Theorem \ref{thm:LowerBound1} is proven via a reduction from the index problem in communication complexity, and for Theorem \ref{thm:LowerBound2} we use a reduction from a variant of the set disjointness problem. The idea behind the proof of Theorem \ref{thm:alg1} is as follows. Consider the following natural and well-known graph sparsification procedure. Given a graph $G$, store every edge, independently of the others, with probability $p$. Let $H$ be the sparsified version of $G$. If $G$ has at least $T$ triangles, the expected number of triangles in $H$ is $p^3 T$. Taking $p = \Omega(T^{-1/3})$, the expected number of triangles in $H$ is $\Omega(1)$, and the number of edges in $H$ is $O(m/T^{1/3})$. The main question is how concentrated is that number? For example, think of two triangles sharing an edge. If this edge was not picked in $H$ then both triangles will not show up in $H$. This phenomenon may translate into a large variance in the number of triangles in $H$. To solve this problem, we identify the graph structure responsible for large variance. More concretely, we call $s$ triangles that share the same edge an \emph{$s$-tower}. For a carefully chosen number $s^*=s^*(p)$, one can show the following fact: If $G$ has no $s^*$-tower, then the variance is small and the number of triangles in $H$ is close to the expectation. If there is an $s^*$-tower, it is tall enough so that at least one floor survives (a floor is two edges that belong to the same triangle). In that case, in the second pass of $A$, the base of that tower is caught, and a triangle is detected.

The algorithm suggested in~\cite{TsourakakisKang} also uses the graph sparsification method. That algorithm computes the sparsified graph $H$, and checks if it contains a triangle. This approach does not allow control over the variance, and ultimately such an algorithm will fail unless it stores $\Theta(m)$ of the edges (think of the case when all triangles are stacked in one tower, then unless $p$ is constant, the base of the tower will be missing almost always). This is also consistent with the lower bound we have for $\triang(1)$ (as computing $H$ can be done in one pass). Our algorithm addresses this issue exactly by trying to either catch a triangle, or catch a floor (which forces the second pass).

\medskip

\noindent \textbf{Paper Organization.} We proceed with the description of our second algorithm mentioned in Section \ref{sec:NewGraphParam}. The proof of Theorem \ref{thm:alg1} follows in Section \ref{sec:ProofAlg1}. The proofs of the lower bounds, Theorems \ref{thm:LowerBound1} and \ref{thm:LowerBound2}, use rather standard techniques, though require some new insights. Both are given in full in Appendix \ref{sec:ProfLowerBound2}.

\section{The second algorithm}\label{apnd:A2}
\begin{figure*}[!htp]
\begin{center}
\fbox{\normalsize
\begin{minipage}{\textwidth}{\textbf{Algorithm $A_2(G,T,\rho(G))$}}
\begin{description}
  \item[Output:] 1 if a triangle is detected, 0 otherwise.

  \item[Pass 1]
  \item[(a)] Sample $4n/\rho(G)$ vertices, uniformly at random. Let $S \subseteq V$ be that set.

  \item[(b)] Store all edges in the steam that touch the set $S$.

  \item[Pass 2] Check for every edge $e$ if it completes a triangle with any of the stored edges. Output 1 iff such edge exists.
\end{description}
\end{minipage}
}
\end{center}
\end{figure*}
\begin{theorem}\label{thm:alg2} Algorithm $A_2$ solves $\dist(2)$ using $O(m/\rho(G))$ bits of memory in expectation.
\end{theorem}
\begin{proof}
Let $Z \subset V$ be the set of vertices in $G$ that belong to some triangle. In our notation, the size of $Z$ is $\rho(G)$. The algorithm never fails if there are no triangles in $G$. Therefore let us consider the case where there are triangles in $G$.

The algorithm $A_2$ fails only if $S \cap Z = \emptyset$. Otherwise, $S$ contains a vertex $v$ that belongs to some triangle $\{v,u,w\}$, and in the first pass the algorithm stores all neighbors of $v$ (and in particular the edges $(v,u)$ and $(v,w)$). In the second pass the edge $(u,w)$ will be considered and $A_2$ will detect the triangle. Let us bound the probability of $S \cap Z = \emptyset$.
Let $A_i$ be the event that the $i^{th}$ vertex chosen to be in $S$ doesn't belong to $Z$. It is easy to see that the $A_i$'s are negatively correlated (as there is no replacement). For every $i$, $Pr[A_i]=1-\rho(G)/n$.
Therefore,
$$Pr[S \cap Z = \emptyset] = Pr[A_1 \wedge A_2 \wedge \dots A_{|S|}] \leq (1-\rho(G)/n)^{4n/\rho(G)} \leq e^{-4}.$$

Now let us compute the expected number of edges stored by $A_2$. For the $i^{th}$ vertex in $S$, let $D_i$ be a random variable counting the degree of that vertex in $G$. Since the $i^{th}$ vertex is a uniformly random vertex, $\E[D_i]=2m/n$ (the average degree in $G$). The expected number of edges touching $S$ is at most (using linearity of expectation)
$$\E\left[\sum_{i=1}^{|S|} D_i\right] = \sum_{i=1}^{|S|} \E\left[D_i\right] = (4n/\rho(G))(2m/n)=8m/\rho(G).$$
\end{proof}

\section{Proof of Theorem \ref{thm:alg1}}\label{sec:ProofAlg1}
We denote by $B(n,p)$ the binomial random variable with parameters $n$ and $p$, and expectation $\mu=np$. We shall use the following variant of the Chernoff bound, whose proof can be found in \cite[p. 21]{JRL2000}. Let $\varphi(x) = (1+x)\ln (1+x) - x$.
\begin{theorem}\label{Thm:Chernoff2} If $X \sim B(n,p)$ and $t \geq 0$ is some number, then
\[ Pr\big(X \geq \mu + t) \leq e^{-\mu  \varphi(t/\mu)}, \qquad Pr\big(X \leq \mu - t) \leq e^{-\mu  \varphi(-t/\mu)}.\]
\end{theorem}

The algorithm $A$ always answers correctly if the graph $G$ has no triangles. Therefore, it suffices to bound the error probability when the graph $G$ has at least $T \geq 1$ triangles.

Let $H$ be the graph in which each edge of the stream is included with probability $p$. Let $\B_1$ be the event that Algorithm $A$ outputs FAIL or the wrong answer, $\B_2$ be the event that more than $5m'$ edges were stored in the first pass (causing the algorithm to output FAIL), and $\B_3$ the event that $H$ has no triangles, and no edge of the stream completes a triangle in $H$. Then
$$Pr[\B_1] \leq  Pr[\B_2] + Pr[\B_3 | \B_2^c]\leq Pr[\B_2] + Pr[\B_3] /Pr[\B_2^c].$$
(in the last inequality we used the fact that for two events $\A,\B$, $Pr[\A|\B] = Pr[\A \wedge \B]/Pr[\B] \leq Pr[\A]/Pr[\B]$).
The number of edges stored by $A$ is a binomial random variable with expectation $mp = m'$. In our case, $m' \geq 4$: we can assume w.l.o.g that $m \geq n/2$ (isolated vertices are never visible to the algorithm), and $T$ always satisfies $T \leq n^3/6$, therefore $m'=6m/T^{1/3} \geq 4$. Using Theorem \ref{Thm:Chernoff2}, the probability
of storing more than $5m'$ edges is at most $1/50$, hence $Pr[\B_2] \leq 1/50$. In turn,
$$Pr[\B_1] \leq \frac{1}{50} + \frac{49}{50}Pr[\B_3].$$
It suffices to show that $Pr[\B_3] \leq 0.3$, and then derive $Pr[\B_1] \leq 1/3$, as required.

We call $s$ triangles that share the same edge an \emph{$s$-tower}. Each pair of edges that belong to the same triangle is called a \emph{floor} in the tower. Let $T_3 \geq T$ be the number of triangles in $G$. For $p=m'/m=6/T^{1/3}$, let $\mu=p^3 T_3 = 216T_3/T$ be the expected number of triangles in $H$ and $\sigma^2$ the variance ($\sigma$ is the standard deviation).
\begin{lemma}\label{lem:noLargeTower} If $G$ contains no tower with more than $T_3^{2/3}$ floors, then $\sigma \leq 110(T_3/T)^{5/6}$.
\end{lemma}
Before we prove this proposition we need the following lemma.
\begin{lemma}\label{lem:CoutingPairs} Let $G$ be a graph with $T_3$ triangles, having no tower with more than $h$ floors. Let $\pi(G)$ be the number of pairs of triangles that share an edge. Then $\pi(G) \leq 3T_3h/2$.
\end{lemma}
\begin{proof}
Observe that every pair of triangles that share an edge belongs to exactly one tower: If the pair belongs to two towers, then the two triangles share two edges, but then they are the same triangle. Every pair belongs to at least one tower, since every such pair is a tower of height two. Therefore we can count the number of pairs sharing an edge, by counting the number of pairs of triangles in every tower.
Let $a_i$ be the number of towers with $i$ floors. Using this notation, $\pi(G) = \sum_{i=2}^{h}a_i \binom{i}{2}.$
Next observe that $\sum_{i=2}^h a_ii \leq 3T_3.$ The sum counts the number of triangles that belong to some tower, when every such triangle is accounted for at most three times (as it belongs to at most three different towers).
Finally, we have
$$\pi(G) = \sum_{i=2}^{h}a_i \binom{i}{2} \leq \frac{1}{2}\sum_{i=2}^{h}(a_ii)\cdot i \leq  \frac{1}{2}\sum_{i=2}^{h}(a_ii)\cdot h
= \frac{h}{2}\sum_{i=2}^{h}a_ii \leq 3T_3h/2.$$
\end{proof}
\begin{proof}(Lemma \ref{lem:noLargeTower})
Index the triangles in $G$ by $1,2,\ldots,T_3$. Let $\one_j$ be the indicator random variable which takes the value 1 if all three edges of triangle $j$ belong to $H$.
$$\E[\one_j]=p^3, \qquad \var[\one_j] = p^3(1-p^3) \leq p^3.$$
In these notations, $\sigma^2$ (the variance of the number of triangles in $G'$) is given by
$$\sigma^2 = \sum_{i=1}^{T_3} \var(\one_i) + \sum_{i < j} \cov(\one_i,\one_j), \qquad \sum_{i=1}^{T_3} \var(\one_i) \leq T_3 p^3 = \frac{216T_3}{T}.$$
For two triangles that share no edge, $\cov(\one_i,\one_j)=\E[\one_i\one_j]-\E[\one_i]\E[\one_j]=p^6-p^6=0.$
Therefore we only need to go over triangles that share an edge. For every such pair,
$\cov(\one_i,\one_j)=p^5-p^6 \leq p^5.$
By Lemma \ref{lem:CoutingPairs} with $h=T^{2/3}$, there are at most $1.5T_{3}^{5/3}$ pairs of triangle that share an edge. Hence,
$$ \sum_{i < j} \cov(\one_i,\one_j) \leq 1.5T_3^{5/3}p^5 = \frac{1.5\cdot 6^5 \cdot T_3^{5/3}}{T^{5/3}}\leq\left(\frac{278T_3}{T}\right)^{5/3}.$$
To summarize,
$$\sigma^2 \leq \frac{216T_3}{T}+\left(\frac{278T_3}{T}\right)^{5/3} \leq \left(\frac{282T_3}{T}\right)^{5/3}.$$
Taking the square root, we get the desired bound on $\sigma$.
\end{proof}
\begin{proposition}\label{prop:ProbNoTower} Conditioned on $G$ \emph{not} having a tower with more than $T_3^{2/3}$ floors, $Pr[\B_3] \leq 0.26$.
\end{proposition}
\begin{proof}
For a random variable $X$, with expectation $\mu$ and standard deviation $\sigma$,  Chebychev's inequality implies $Pr[X=0] \leq  \left(\frac{\sigma}{\mu}\right)^2.$
The expected number of triangles in $H$ is $\mu = 216T_3/T$. The standard deviation $\sigma \leq 110(T_3/T)^{5/6}$ (by Lemma \ref{lem:noLargeTower}).
Therefore
$$Pr[\text{no triangles in $H$}] \leq \left( \frac{110}{216}\cdot \left(\frac{T}{T_3}\right)^{1/6} \right)^2 \leq 0.26.$$
\end{proof}

Next we turn to the case where $G$ contains a tower with at least $T_3^{2/3}$ floors.

\begin{proposition}\label{prop:ProbLargeTower} Conditioned on  $G$ having a tower with at least $T_3^{2/3}$ floors, $Pr[\B_3] \leq 0.001$.
\end{proposition}
\begin{proof}
Fix a tower with at least $T_{3}^{2/3}$ floors. Every floor belongs to $H$ independently of the others with probability $p^2$. Therefore the expected number of floors that belong to $H$ from that tower is
$$p^2T_3^{2/3}=\left(\frac{6}{T^{1/3}}\right)^2 T_3^{2/3} = 36\left(\frac{T_3}{T}\right)^{2/3} \geq 36.$$
Using Chernoff's bound (second inequality of Theorem \ref{Thm:Chernoff2}) with $\mu = 36$ and $t=35$, we get
$$Pr[\text{no floor from the tower belongs to $H$}] \leq e^{-36f(-36/35)} \leq e^{-30} \leq 0.001.$$
\end{proof}

Finally, from Propositions \ref{prop:ProbNoTower} and \ref{prop:ProbLargeTower} we get $Pr[\B_3] \leq 0.26 + 0.001 < 0.3.$
To complete the proof of Theorem \ref{thm:alg1}, observe that the space complexity of $A$ never exceeds $5m'$ which is $30m/T^{1/3}$.

\appendix

\section{Proof of Theorem \ref{thm:LowerBound1}}\label{sec:ProfLowerBound1}
The theorem is a direct corollary of the same lower bound, just for the problem $\dist(1)$. Clearly, if one can approximate the number of triangles, one can distinguish between the case where there are no triangles, or at least $T$  triangles.

\begin{proposition}\label{prop:LowerBoundDist1} $\exists c_1,c_2 > 0$ s.t. the space complexity of $\dist(1)$ is $\Omega(m)$, when the input is an $n$-vertex graph with $m \in [c_1n,c_2n^2]$ edges. Furthermore, this lower bound is true even if the graph has as many as $0.99n$ triangles.
\end{proposition}
In the rest of this section we prove Proposition~\ref{prop:LowerBoundDist1}. It suffices to consider only the case $T=0.99n$ (since the way $\dist$ is defined, $T$ is a lower bound on the number of triangles).

Let $\ind_p$ be the following problem: Alice has a binary vector $\x$ of length $p$ and Bob has an index $\ell$ between 1 and $p$. Alice communicates with Bob (one way communication) in order to determine the value of $\x[\ell]$ with probability better than $1/2$. The randomized communication complexity of this problem is $\Omega(p)$ \cite{IndexProb10,JKS06}. For a fixed function $g(\cdot)$, let $\G(n,g(n))$ be a family of graphs on $n$ vertices with $\Theta(g(n))$ edges and $0.99n$ triangles. Assume by contradiction that there exist $g(\cdot)$, and an algorithm $A$ that solves $\dist(1)$ for the graphs in $\G(n,g(n))$ using $o(g(n))$ memory. We shall use this algorithm to solve $\ind_p$ using $o(p)$ bits of communication, deriving a contradiction to the established lower bound for that problem.

Consider the following graph $G(\x,\ell)$. Let $f(n)=g(n)/n$, and let $a$ be such that $af(a)=p$ (assume w.l.o.g that $f(a) \leq a$, otherwise rename them). The vertex set of $G$ consists of $n$ vertices partitioned into three sets: $V=X \cup Y \cup
Z$, $|X|=a,|Y|=f(a),|Z|=T$. We require $T = 0.99n$, and $a+f(a)+T =n$, therefore $n=(a+f(a))/(1-0.99)= \Theta(a)$.
Let $x_i$ be the $i^{th}$ vertex in $X$ (and similarly $y_i$ in $Y$, and $z_i$ in $Z$). Define the edge
set $E_1$ as follows: the first $f(a)$ entries in $\x$ determine the
neighbors of $x_1$ in $Y$ (we place the edge $(x_1,y_j)$ iff $\x[j]=1$), the next $f(a)$ entries determine the neighbors of
$x_2$, and so on. Define the edge set $E_2$ as follows: let $e_\ell=(x_i,y_j)$
be the edge corresponding to Bob's index $\ell$ in $\x$; add $2T$ edges of the
form $(z_r,x_i)$ and $(z_r,y_j)$ for $r=1,\dots,T$. Finally, $E(G)=E_1 \cup E_2$. Let $m=|E(G)|$.

The graph $G$ enjoys the following properties: $(a)$ $G \in \G(n,g(n))$: the number of edges in $G$ is $m=\Theta(p)=\Theta(af(a))$, and since $n=\Theta(a)$, $m=\Theta(nf(n))=\Theta(g(n))$, $(b)$ the graph $G$ has $T$ triangles if $\x[\ell]=1$ and no triangles otherwise.

To solve $\ind_p$, Alice feeds the algorithm $A$ with $E_1$, records its memory tape. When finished, she sends it to Bob. Bob then feeds $A$ with the edge set $E_2$, and answers according to $A$. The correctness is now immediate: since $G(\x,\ell) \in \G(n,g(n))$ the algorithm $A$ answers correctly with probability at least $2/3$, and by property $(b)$, this is also the correct answer for $\ind_p$. As for the communication complexity, the number of edges in $G$ satisfies $m=\Theta(p)$ (since $\x$ contains $\Theta(p)$ ones, and Bob adds only $2T=O(p)$ edges). Algorithm $A$ uses $o(m)$ bits of memory, therefore the data sent by Alice is of the order of $o(m)=o(p)$. Contradiction is then derived. Finally observe that $g(n)$ can be arbitrary, therefore $m=\Theta(g(n))$ has the desired range $[c_1n,c_2n^2]$.

\section{Proof of Theorem \ref{thm:LowerBound2}}\label{sec:ProfLowerBound2}
The theorem is a direct corollary of the same lower bound, just for the problem $\dist(O(1))$. Clearly, if one can approximate the number of triangles, one can distinguish between the case where there are no triangles, or at least $T$  triangles.

\begin{proposition}\label{prop:LowerBoundDist2} The space complexity of $\dist(O(1))$ is $\Omega(m/\max\{T_3,1\})$, for input graphs with $m$ edges and $T_3$ triangles.
\end{proposition}

In the rest of this section we prove Proposition~\ref{prop:LowerBoundDist2}. Let $\disj_p^r$ be the following problem: Alice and Bob have each a vector of length $p$ with exactly $r$ ones in each vector. Each vector is interpreted as the characteristic vector of a subset in $\{1,2,\ldots,p\}$. Alice and Bob communicate in order to decide whether their sets intersect or not. Let us define the problem $\disj_p^{r,t}$ to be the same as $\disj_p^r$ just that now the intersection is promised to be either empty or of size at least $t$. Observe that the size of the intersection is also given by $\sum_{i=1}^{p}\x_i \y_i$ ($\x,\y$ are their vectors).

We first describe a reduction from $\disj_p^{r,t}$ to $\dist(O(1))$, then establish a lower bound on the communication complexity of $\disj_p^{r,t}$. Let $\x,\y$ be two vectors in $\{0,1\}^p$ for $p=n^2$, $n$ an integer (w.l.o.g. we can assume $n$ is an integer, since we can always pad the vectors $\x$ and $\y$ with zeros). Consider the following graph $G^*=G(\x,\y)$. The set of vertices $V=A \cup B \cup C$, each part of size $n$. Let $a_i$ be the $i^{th}$ vertex in $A$ (and similarly define $b_i,c_i$). We interpret the vector $\x$ as follows: the first $n$ entries in $\x$ determine the neighbors of $a_1$ in $B$, $(a_1,b_j)\in E$ iff  $\x[j]=1$. In the same way, entries $[(i-1)n,in-1]$, for $i=2,\ldots,n$ determine the neighbors of $a_i$. Similarly, $\y$ determines the neighbors of $c_i$ in $B$ for $i=1,\ldots,n$.
In addition, we have the following set of $n$ edges: $(a_i,c_i) \in E$ for $i=1,\ldots,n$ (a perfect matching on $A$ and $C$).

\begin{lemma}\label{lem:ReductionNumTriang} The graph $G^*$ has $T$ triangles iff $\sum \x_i\y_i = T$.
\end{lemma}
\begin{proof} Consider a triangle in $G^*$, it must contain exactly one edge of the perfect matching, $(a_i,c_i)$.
To complete the triangle there must be two additional edges $(a_i,b_k),(c_i,b_k)$. This however implies that $\x_{(i-1)n+k}=\y_{(i-1)n+k}=1$.
If on the other hand $\x_t=\y_t=1$ for $t=(i-1)n+k$ for $i,k \in [n]$, then the edges $(a_i,b_k),(c_i,b_k)$ are present in $G^*$. Therefore together with the edge $(a_i,c_i)$ they induce a triangle in $G^*$.
\end{proof}

\begin{lemma}\label{lem:ComplxtyOfDisj} The communication complexity of $\disj_p^{r,t}$ is $\Omega(r/t)$ for any $r < p/2$.
\end{lemma}
\begin{proof} Assume by contradiction that the communication complexity of $\disj_p^{r,t}$ is $o(r/t)$. Let $r'=r/t$, and consider the problem
$\disj_p^{r'}$. Given two vectors $\x',\y'$ of size $r'$, we construct the vector $\x$ by taking $t$ concatenated copies of $\x'$. Similarly construct $\y$. Clearly, if $\x',\y'$ intersect then $\x,\y$ have intersection size at least $t$, and if $\x',\y'$ are disjoint so are $\x,\y$. We can then solve $\disj_p^{r'}$ using $o(r/t)=o(r')$ bits of communication by reducing to $\disj_p^{r,t}$. This however contradicts the lower bound established in \cite{HastadWigderson} on $\disj_p^{r'}$.
\end{proof}

\begin{proposition}\label{prop:Reduction} If there exists an algorithm $A$ that solves $\dist(O(1))$ using $o(m/T)$ bits of memory, then there exists an algorithm $A^\ast$ that solves $\disj_p^{r,t}$ using $o(r/t)$ bits of communication whenever $r = \Omega(\sqrt{p})$.
\end{proposition}
\begin{proof}
We describe the algorithm $A^*$ for $\disj_p^{r,t}$. Alice has the vector $\x$ and Bob has $\y$. Alice runs $A$ on the stream of edges of $G^*$ that include the matching edges and the edges induced by $\x$. She then sends the content of her memory to Bob, who continues to run $A$, while feeding it the edges induced by $\y$.  At the end, Bob sends the content of his memory to Alice, and this repeats for the number of passes that $A$ requires (which is constant). They answer `Disjoint' iff $A$ outputs 0. The correctness of the algorithm comes from Lemma \ref{lem:ReductionNumTriang}: if $\x$ and $\y$ are disjoint, then $G^*$ has no triangles, and $A$ outputs 0 with probability at least $2/3$. If $\x$ and $\y$ intersect, then the intersection size is at least $t$, and hence $G$ has at least $T=t$ triangles. Accordingly, $A$ outputs 1 with probability at least $2/3$. The communication complexity of $A^*$ is the same, up to a constant factor, as the memory used by $A$ (as only the memory content is being transmitted). The number of edges $m$ in $G^*$ is $\Theta(r)$: $n$ matching edges, and $2r$ edges coming from $\x$ and $\y$. Since $r = \Omega(\sqrt{p})=\Omega(n)$ we have $m=\Theta(r)$. The number of triangles in $G^*$ is $T=t$, therefore, $A^*$ solves $\disj_p^{r,t}$ using $o(m/T)=o(r/t)$ bits of communication.
\end{proof}
Proposition \ref{prop:LowerBoundDist2} follows from Lemma~\ref{lem:ComplxtyOfDisj} and Proposition~\ref{prop:Reduction}.

\end{document}